\documentclass[sigconf]{acmart}
\usepackage{amsmath,amssymb,amsfonts}
\usepackage{amsthm}
\newtheorem{theorem}{Theorem}
\newtheorem{lemma}{Lemma}
\usepackage{algorithm}
\usepackage{algorithmicx}
\usepackage{algpseudocode}
\usepackage{natbib}
\usepackage{graphicx}
\usepackage{textcomp}
\usepackage{tikz}
\usepackage{verbatim}
\usepackage[tightpage]{preview}
\PreviewEnvironment{tikzpicture}
\usepackage{xcolor}
\AtBeginDocument{%
  \providecommand\BibTeX{{%
    \normalfont B\kern-0.5em{\scshape i\kern-0.25em b}\kern-0.8em\TeX}}}

\begin{document}
\fancyhead{}
\settopmatter{printacmref=false} 
\renewcommand\footnotetextcopyrightpermission[1]{} 

\title{Asynchronous Byzantine Agreement\\ in Incomplete Networks\\}

\author{Ye Wang}
\email{wangye@ethz.ch}
\affiliation{%
  \institution{ETH Zurich}
  \city{Zurich}
  \country{Switzerland}
}

\author{Roger Wattenhofer}
\email{wattenhofer@ethz.ch}
\affiliation{%
  \institution{ETH Zurich}
  \city{Zurich}
  \country{Switzerland}
}


\begin{abstract}

The Byzantine agreement problem is considered to be a core problem in distributed systems. For example, Byzantine agreement is needed to build a blockchain, a totally ordered log of records. 
		Blockchains are asynchronous distributed systems, fault-tolerant against Byzantine nodes. 
		
		In the literature, the asynchronous byzantine agreement problem is studied in a fully connected network model where every node can directly send messages to every other node. This assumption is questionable in many real-world environments. In the reality, nodes might need to communicate by means of an incomplete network, and Byzantine nodes might not forward messages. Furthermore, Byzantine nodes might not behave correctly and, for example, corrupt messages. Therefore, in order to truly understand Byzantine Agreement, we need both ingredients: asynchrony and incomplete communication networks.
		
		In this paper, we study the asynchronous Byzantine agreement problem in incomplete networks. A classic result by Danny Dolev proved that in a distributed system with $n$ nodes in the presence of $f$ Byzantine nodes, the vertex connectivity of the system communication graph should be at least $(2f+1)$. While Dolev's result was for synchronous deterministic systems, we demonstrate that the same bound also holds for asynchronous randomized systems. We show that the bound is tight by presenting a randomized algorithm, and a matching lower bound. This algorithm is based on a protocol which allows other Byzantine agreement algorithms to be implemented in incomplete networks.
		
		\textit{Index Term}--- Blockchain, Byzantine agreement, communication network, randomized algorithms 
  
\end{abstract}
\maketitle





\pagestyle{plain}
\pagenumbering{arabic}
\section{Introduction}
	
	
	
	
	Byzantine agreement is at the heart of understanding distributed systems.
	Most existing work about byzantine agreement assumes a fully connected network, i.e., every node in the distributed system can directly communicate with every other node. 
	In reality, however, nodes are often connected by an unreliable network, such as the Internet. To communicate, two nodes must exchange messages, and these messages will be forwarded by relay nodes, which are controlled by third parties. Relay nodes may be compromised, even Byzantine. A relay node may for instance decide to corrupt or simply drop messages. This is particularly true in world-scale distributed systems, and world-scale systems (blockchains) are predominantly responsible for the current reawakened interest in byzantine agreement.
	
	
	 The Byzantine agreement problem has also been studied \textit{in a network}. Already as early as 1982, Danny Dolev \cite{dolev1982byzantine}  showed that two conditions are both necessary and sufficient to achieve Byzantine agreement in an $n$-node system: As usual, we need the number of Byzantine nodes $f$ to obey $f<\frac{n}{3}$; In addition, the vertex connectivity of the communication graph cannot be less than $(2f+1)$.
	
	Dolev's fundamental work had only considered the synchronous model, i.e., all communication happened in synchronous rounds. Moreover, the result was restricted to deterministic algorithms. 
	
	In world-scale systems, it is difficult to argue for synchronous communication. Consequentially, the focus has shifted away from synchronous systems towards asynchronous systems. Even though some form of synchrony is usually needed for liveness, safety is guaranteed even in completely asynchronous systems. This is the case for permissionless blockchain systems such as Bitcoin\cite{nakamoto2019bitcoin}, 
	and also for permissioned systems such as PBFT\cite{castro1999practical}. 
	

	\begin{center}
		\begin{figure}[!htb]
			\centering
			\includegraphics[width=8.5cm]{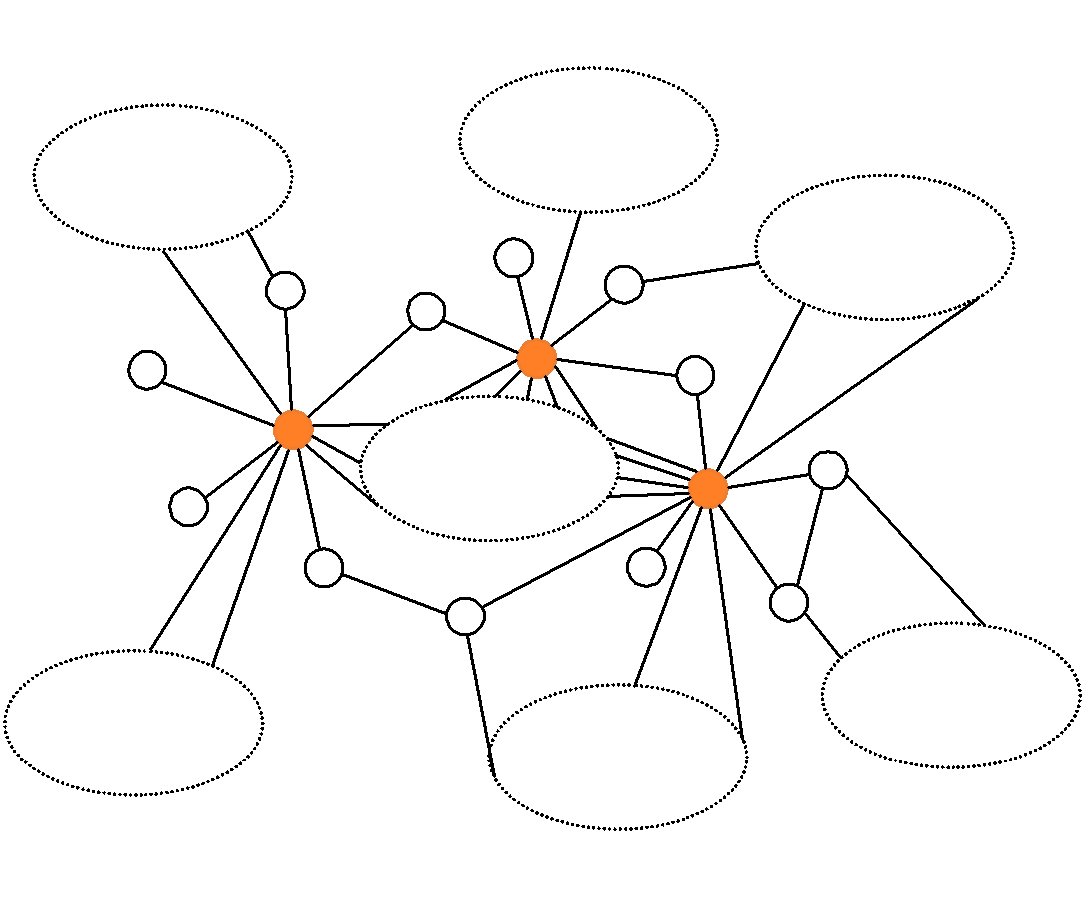}
			\caption{An incomplete network with 7 large node clusters and 14 additional nodes. The 3 solid nodes are well-connected; if these 3 nodes are Byzantine, there may split the graph into small disconnected subgrpahs.}
		\end{figure}
	\end{center}

	In this paper, we study the effect of \textit{incomplete networks} on byzantine agreement in \textit{asynchronous} distributed systems. We have a system with $n$ nodes, each with a binary input. At most $f$ nodes, with $f<\frac{n}{3}$, are Byzantine. The nodes are interconnected by an incomplete network, which can be represented by an undirected graph $G=(V_G,E_G)$. An example of such an incomplete network is shown in Figure 1.
	
	It was proven in \cite{fischer1985impossibility} that there is no deterministic algorithm which can solve Byzantine agreement in an asynchronous distributed system, with $f=1$, i.e., a single Byzantine node, or even a crash node. However, randomization helps \cite{ben1983another}, and consequently our algorithm will be randomized as well. More precisely,
	we present a randomized algorithm that solves the Byzantine agreement problem in asynchronous distributed systems with high probability, as long as the vertex connectivity of the network $G$ is at least $(2f+1)$. We also show that this bound is tight by proving that there is no randomized algorithm which can solve the asynchronous Byzantine agreement problem when the vertex connectivity of $G$ is less than $(2f+1)$. In other words, we demonstrate that with the help of randomization, Dolev's original bounds \cite{dolev1982byzantine} also hold in asynchronous networks.
	
	The rest of this paper is organized as follows: In the next section, we give an overview of related work. Section 3 defines the model that we study in this paper. In Section 4, we present a randomized algorithm for solving the Byzantine agreement in our model. We give the necessary condition for solving the Byzantine agreement problem under the model in Section 5. We conclude our work in Section 6.
	
	\section{Related Work}
	
	In this section, we discuss all the known related work on Byzantine agreement which are studied in asynchronous distributed systems and incomplete networks.
	
	The Byzantine agreement problem is essential in distributed systems. It was proved in \cite{fischer1985impossibility} that there is no deterministic algorithm that can solve the Byzantine agreement problem in the presence of even a single Byzantine node. Later Lamport et al. \cite{lamport2019byzantine} proved that if the number of Byzantine nodes is larger than or equal to $\frac{n}{3}$, it is impossible for any algorithm to solve the Byzantine agreement problem with $n$ nodes. The first randomized algorithm was proposed by Ben-Or \cite{ben1983another}. With this algorithm, all correct nodes will decide on the same output with high probability even when a constant fraction of nodes is faulty. Bracha \cite{bracha1987asynchronous} improved the previous result and proposed a randomized algorithm that solves the asynchronous Byzantine agreement problem if less than one third of nodes are Byzantine.
	
	The connectivity conditions of communication graphs have been studied right when the topic emerged. Danny Dolev \cite{dolev1982byzantine} studied Byzantine agreement in synchronous distributed systems among $n$ nodes in the presence of $f$ Byzantine nodes and figured out that it is possible to achieve Byzantine agreement in networks with vertex connectivity at least $(2f+1)$ and $n\geq 3f+1$. Later Dolev et al. and Abraham et al.\cite{dolev1986reaching,abraham2004optimal} showed that the approximate Byzantine agreement problem can be solved in synchronous networks if and only if $(2f+1)$ vertex connectivity is given.

	Recently, there are a series of results under the local broadcast model. In contrast to the orthodox point-to-point communication model, under the local broadcast model, all neighbors of a transmitting node are guaranteed to receive identical messages. A lower connectivity requirement under the local broadcast model in the synchronous setting was obtained in \cite{khan2019exact}. In the presence of $f$ Byzantine nodes, the following conditions are both necessary and sufficient. The communication graph $G$ with $n$ nodes has minimum degree $2f$ and $G$ is $( \left \lfloor \frac{3f}{2}+1 \right \rfloor)$-connected. The local broadcast model in the asynchronous distributed system was considered in \cite{samir2019asynchronous}. They show that it is necessary to have $(2f+1)$ vertex connectivity for solving the approximate Byzantine agreement problem. This bound keeps the same as the bound of the point-to-point communication model given in \cite{dolev1986reaching}.

\begin{figure*}
\centering
\begin{tikzpicture}[
    level/.style={rectangle, draw=none, rounded corners=1mm, fill=white, text centered, anchor=north, text=black},
    state/.style={rectangle, draw=none, fill=white, text width=2cm,
        text centered, anchor=north, text=black},
    leaf/.style={rectangle, draw=none, fill=white,
        text centered, anchor=north, text=black}, text width=3cm,
    level distance=0.4cm, growth parent anchor=south
]
\node (State00) [level] {Byzantine Agreement in Incomplete Networks} [->]
        [sibling distance = 8.6cm, level distance=1.5cm]
        child{[sibling distance=5.5cm]
        node (State01) [state] {Exact Agreement}
        child{[sibling distance=2.8cm, level distance=1.2cm]
        node (State03) [state] {Synchronous Network}
        child{[level distance=1cm]
        node (State11) [state] {Deterministic Algorithm}
        child{[level distance=1cm]
        node (State05) [state] {Point-to-Point}
        child{[level distance=1cm]
        node (State08) [state] {Necessary\& Sufficient Conditions}
        child{
        node (State10) [leaf] {Danny Dolev \cite{dolev1982byzantine}}
        }
        }
        }
        child{[level distance=0.61cm]
        node (State05) [state] {Local Broadcast}
        child{[level distance=1cm]
        node (State08) [state] {Necessary\& Sufficient Conditions}
        child{
        node (State10) [leaf] {Khan et al.\cite{khan2019exact}}
        }
        }
        }
        }
        }
        child{[level distance=1.2cm, sibling distance=2.86cm]
        node (State04) [state] {Asynchronous Network}
        child{[level distance=4.65cm]
        node (State11) [state] {Deterministic Algorithm}
        child{
        node (State12) [leaf] {Impossible Lamport et al.\cite{lamport2019byzantine}}
        }
        }
        child{[level distance=1cm]
        node (State05) [state] {Randomized Algorithm}
        child{[level distance=1cm]
        node (State07) [state] {Point-to-Point}
        child{[level distance=1cm]
        node (State08) [state] {Necessary\& Sufficient Conditions}
        child{
        node (State10) [leaf] {Our Work}
        }
        }
        }
        }
        }
        }
        child{[sibling distance=3cm]
        node (State02) [state] {Approximate Agreement}
        child{[level distance=1.2cm]
        node (State04) [state] {Synchronous Network}
        child{[level distance=1cm]
        node (State06) [state] {Deterministic Algorithm}
        child{[level distance=1cm]
        node (State07) [state] {Point-to-Point}
        child{[level distance=1cm]
        node (State08) [state] {Necessary\& Sufficient Conditions}
        child{
        node (State10) [leaf] {Dolev et al.\cite{dolev1986reaching} Abraham et al.\cite{abraham2004optimal}}
        }
        }
        }
        }
        }
        child{[level distance=1.2cm]
        node (State03) [state] {Asynchronous Network}
        child{[level distance=1cm]
        node (State06) [state] {Deterministic Algorithm}
        child{[level distance=1cm]
        node (State07) [state] {Local Broadcast}
        child{[level distance=1cm]
        node (State08) [state] {Necessary Condition}
        child{
        node (State10) [leaf] {Khan et al.\cite{samir2019asynchronous}}
        }
        }
        }
        }
        }
        }

;
        
\end{tikzpicture}
\caption{Comparison with previous work.}
\end{figure*}
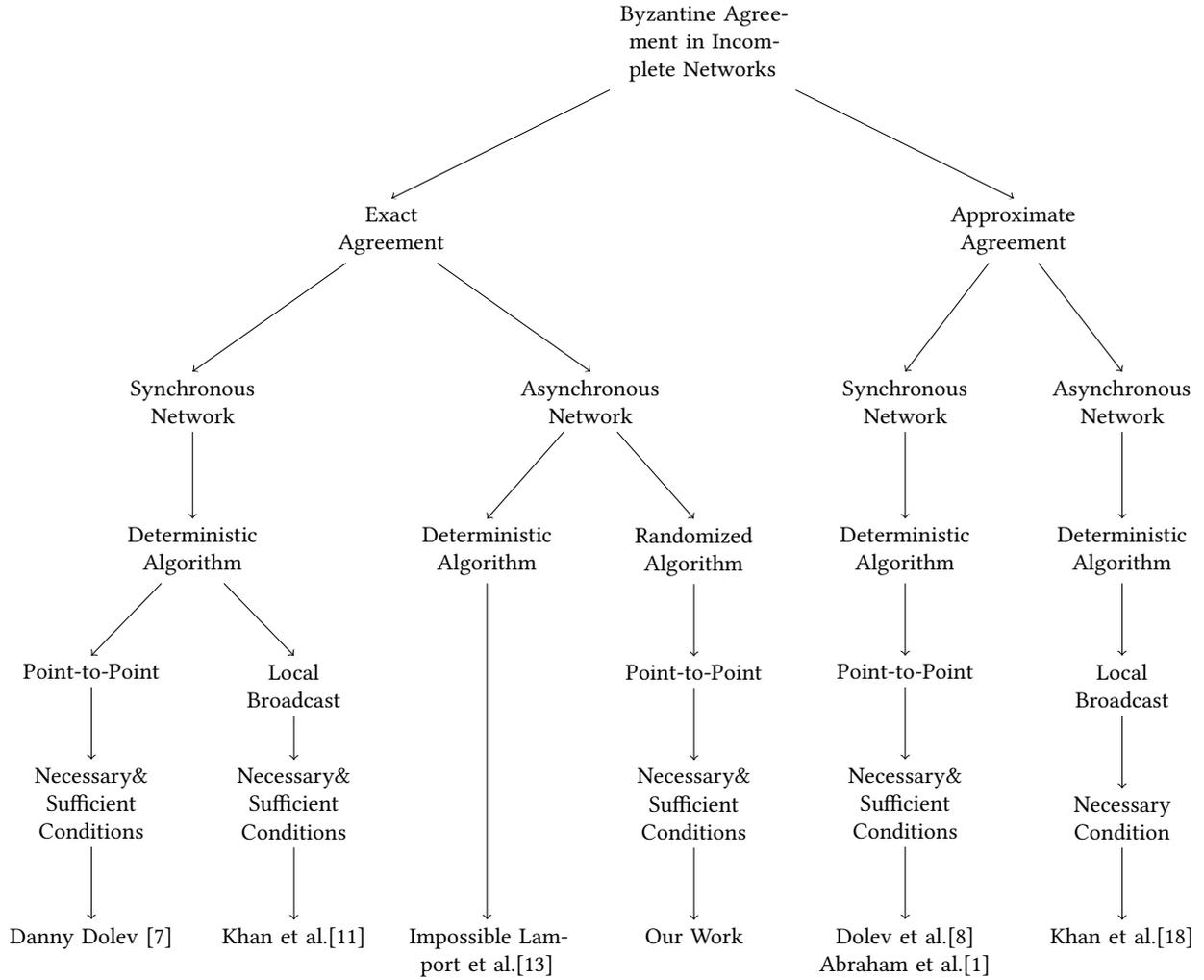

	 We are the first to study the exact Byzantine agreement problem in asynchronous distributed systems, as the comparison between our work and related work shown in Figure 2. In contrast to results in \cite{dolev1986reaching, samir2019asynchronous}, which consider approximate Byzantine agreement, we study the exact Byzantine agreement problem. To obtain an approximate Byzantine agreement, nodes are only required to obtain values that are close to each other, rather than identical. However, in many real-world systems, such as financial systems, one requires precise values. Any minor deviation may accumulate and cause serious consequences. Other work\cite{dolev1982byzantine, khan2019exact} solved exact Byzantine agreement solely in synchronous distributed systems by deterministic algorithms. Unfortunately, assuming synchrony is often not realistic in many world-scale systems. The impossibility for deterministic algorithms to solve the asynchronous Byzantine agreement problem inspires us to propose a randomized algorithm that achieves Byzantine agreement on undirected incomplete $(2f+1)$-connected communication graphs in asynchronous distributed systems with $n$ nodes in the presence $f$ Byzantine nodes.
	
	\section{Model and Notation}
	
	\subsection{Distributed System}
	
	Given a distributed system with $n$ nodes, the communication network between nodes is represented by an undirected graph $G=(V_G, E_G)$. Nodes $u$ and $v$ can send messages to each other if and only if they are adjacent in $G$, i.e., $(u,v)\in E_G$. We will use the terms of \textit{node} and \textit{vertex} interchangeably in this paper. The communication channels between nodes are authenticated. Nodes can recognize who is the sender of the message when they receive a message. Messages cannot be modified by any third party if messages are delivered via the authenticated channel between adjacent nodes. However, relay nodes on paths between nodes have the ability to modify messages or even generate fake messages. 
	
	\subsection{System parameters}
	
	The system we study in this paper has two critical parameters.
	
	\textbf{Asynchronous systems or synchronous systems}: In a synchronous system, there is an upper bound on the message delivery delay from one node to another which is known by all nodes. Nodes take actions in rounds. Each round takes a constant period of time which is sufficient for nodes to send messages, do local computations and accept incoming messages. In asynchronous distributed systems, the message delay from one node to another has no finite upper bound. Messages may be delayed for arbitrary time periods. However, messages will eventually be delivered. Nodes take actions when they are activated by events, such as messages arriving. In this paper, we study asynchronous systems.
		
		
	\textbf{Broadcast transmission or point-to-point transmission}: if the transmission mechanism is point-to-point, nodes can send a message to at most one neighbor at a time. If the transmission mechanism is broadcast, node $u$ sends identical messages to all of its neighbors at the same time. In this paper we study the point-to-point transmission mechanism.

	\subsection{Byzantine agreement}
	
	There are Byzantine nodes in distributed systems, which may behave arbitrarily. A Byzantine node may corrupt or simply drop messages. The system requires agreement among all nodes in the system, e.g. a total order of blocks in blockchain systems. Because there are general reduction protocols from multivalued agreement to binary agreement\cite{turpin1984extending, mostefaoui2000binary}, we focus on the Byzantine binary agreement problem in this paper.
	
	Every node has a binary input, and we want correct nodes decide on a binary value which satisfies the following conditions,
	
	\textbf{Agreement (exact)}:	all correct nodes decide for the same value.
		
	\textbf{Agreement (approximate)}: for any preassigned $\epsilon>0$, all correct nodes decide with outputs that are within $\epsilon$ of each other.
		
	\textbf{Termination (deterministic)}: all correct nodes terminate in a finite time.
		
	\textbf{Termination (probabilistic)}: the probability that a correct node is undecided after $r$ steps approaches $0$ as $r$ approaches infinity.
		
	\textbf{Validity}: the decision value must be the input value of a node.
	
	In contrast to the exact agreement, the approximate agreement is not valid for many real-world distributed systems. Thus, we focus on the exact agreement problem in this paper.

	Because it is impossible for any deterministic algorithm to solve exact Byzantine agreement problem in asynchronous distributed systems in the presence of a single faulty node, we study randomized algorithms with the probabilistic termination condition. There exists a source that generates random numbers in randomized algorithms. Nodes have access to these random numbers during executions. The agreement made by nodes is based on their input and these random values. Therefore, with the same input, the behavior of nodes can be different because of these random numbers.
	
	\subsection{Connectivity}
	
	A \textbf{path} $P$ in an undirected graph $G=(V_G, E_G)$ is a finite sequence of edges which joins a sequence of distinct vertices. A $uv$-path $P_{uv}$ is a path between nodes $u$ and $v$. Nodes $u$ and $v$ are endpoints of $P_{uv}$. Nodes other than $u$ and $v$ in $P_{uv}$ are internal nodes of $P_{uv}$.
	
	If there is a path between every node pair $u, v\in V_G$, then $G$ is connected. $G$ is $k$-connected if $G$ is still connected after removing $k-1$ arbitrary nodes. There is a classic result for $k$-connected graphs:
	
	\begin{theorem}[Menger's Theorem\cite{menger1927allgemeinen}]
		An undirected graph $G=(V_G, E_G)$ is $k$-connected if and only if for any two nodes $u, v\in V_G$, there are $k$ node disjoint $uv$-paths. Two $uv$-paths are disjoint if and only if they do not have any identical internal node.
	\end{theorem}
	
	

	\section{Sufficient Conditions}
	
	In this section, we present an algorithm that solves the asynchronous Byzantine agreement problem with high probability on $G=(V_G, E_G)$, when $G$ is $(2f+1)$-connected. We will introduce this algorithm in three levels, top to bottom. Lower level protocols provide fundamental functions for constructing higher level protocols.

	The primary protocol is discussed in Subsection 4.1, which is the top level of our algorithm. It is an $f$-resilient agreement protocol that describes the behavior of nodes, including how nodes communicate with each other and how they decide on the agreement.
	
	The middle level of our algorithm is a broadcast algorithm, as the primary $f$-resilient agreement protocol above requires a broadcast communication mechanism. However, our model is built on the point-to-point communication mechanism, which does not meet the requirement. We use an authenticated double-echo broadcast algorithm to fix this gap. We discuss this broadcast algorithm in Subsection 4.2. With this broadcast algorithm, identical messages will be delivered to all nodes.
	
	The bottom level of our algorithm is an algorithm that allows any node pair to communicate with each other in incomplete networks. The middle level algorithm, i.e., the authenticated double-echo algorithm is implemented in systems where nodes can directly communicate with each other. Therefore, we design an asynchronous purifying algorithm to effectuate the broadcast algorithm in our model. In Subsection 4.3, Subsection 4.4 and Subsection 4.5, we present the asynchronous purifying algorithm, which ensures that the message delivery between any node pair is correct in an incomplete communication network with the presence of Byzantine nodes. We explain how nodes send messages, transmit messages as internal nodes on a path and accept messages respectively.

	\subsection{Byzantine agreement Algorithm}
	
	The top level protocol is an $f$-resilient agreement protocol in asynchronous systems in the presence of $f<\frac{n}{3}$ Byzantine nodes, which is derived from the probabilistic protocol in \cite{bracha1987asynchronous}. In asynchronous systems, there is no global real-time clock. Nodes take actions when they receive messages from other nodes.
	
	\begin{algorithm}[!htb]
		\caption{Code for node $u$, $Phase(i)$, $i=0, 1, ...$}
		\begin{algorithmic}[1]
			\renewcommand{\algorithmicfunction}{\textbf{round 1}}
			\Function{}{}
			\State $\textbf{broadcast} (\mbox{source} = u, \mbox{round} = 3i+1, \mbox{value} = v_u)$
			\State $\textbf{wait}$ until validate $(n-f)$  messages of round $3i+1$
			\If{more than $\frac{n-f}{2}$ messages of round $3i+1$ \\ \ \ \ \ \ \ \  have the same value $v_m$}
			\State $v_u\gets v_m$
			\EndIf
			\EndFunction
			\State
			\renewcommand{\algorithmicfunction}{\textbf{round 2}}
			\Function{}{}
			\State $\textbf{broadcast}\ (\mbox{source} = u, \mbox{round} = 3i+2, \mbox{value} = v_u)$
			\State $\textbf{wait}$ until validate $(n-f)$ messages of round $3i+2$ 
			\If {more than $\frac{n}{2}$ messages have the same value $v_m$ \\ \ \ \ \ \ \ \ other than $v_u$}
			\State $v_u\gets v_m$
			\State $ready \gets True$
			\Else
			\State $ready \gets False$
			\EndIf
			\EndFunction
			\State
			\renewcommand{\algorithmicfunction}{\textbf{round 3}}
			\Function{}{}
			\If {$ready$}
			\State $\textbf{broadcast}\ (\mbox{source} = u, \mbox{round} = 3i+3,$\\ \ \ \ \ \ \ \ \ \ \ \ \ \ \ \ \ \ \ \ \ \ \ \ $\mbox{value} = v_u)$
			\Else
			\State $\textbf{broadcast}\ (\mbox{source} = u, \mbox{round} = 3i+3,$\\ \ \ \ \ \ \ \ \ \ \ \ \ \ \ \ \ \ \ \ \ \ \ \ $ \mbox{value} = \emptyset)$
			\EndIf
			\State $\textbf{wait}$ until validate $(n-f)$ messages of round $3i+3$
			\If {more than $2f$ messages have the same value \\ \ \ \ \ \ \ \ $v_m\neq \emptyset$}
			\State $v_u\gets v_m$
			\State $decision_u\gets v_u$
			\For{$j \gets 1$ to $3$}
			\State $\textbf{broadcast}\ (\mbox{source} = u, \mbox{round} = 3(i+1)+j,$\\ \ \ \ \ \ \ \ \ \ \ \ \ \ \ \ \ \ \ \ \ \ \ \ \ \ \ \ $ \mbox{value} = v_u)$
			\EndFor
			\State \textbf{terminate}
			\Else
			\If {more than $f$ messages have the same value \\ \ \ \ \ \ \ \ \ \ \ \ \ $v_m$}
			\State $v_u\gets v_m$
			\Else
			\State $v_u\gets coin\_toss$ (0 or 1 with probability $\frac{1}{2}$)
			\EndIf
			\EndIf
			\EndFunction
			\State
			\State \textbf{go to} $Phase(i+1)$
		\end{algorithmic}
	\end{algorithm}
	
	Nodes take actions in phases. We present the protocol of phase $i$ in Algorithm 1. The algorithm starts in phase 0. The initial value of $v_u$ is the input of node $u$ and $decision_u$ is initialized as $\perp$. We assume that nodes communicate with each other by reliable broadcast channels in this subsection. If a message is delivered to one node, it is also delivered to other nodes in the system.
	
	In each phase, node $u$ takes actions in three rounds. Node $u$ keeps a value at the beginning of each phase and broadcasts this value to other nodes in the first round. After validating $(n-f)$ messages of the first round from other nodes, node $u$ sets its value as the majority of these $(n-f)$ messages and enters in the second round. In the second round, node $u$ broadcasts the new value and waits until validating $(n-f)$ messages of the second round from other nodes. If there is a value which is accepted in more than $\frac{n}{2}$ messages, then $u$ sends this value at the beginning of the third round. Otherwise, node $u$ sends an empty message in the third round. Node $u$ waits until validating $(n-f)$ messages of the third round from other nodes. If more than $2f$ messages have the same value, then $u$ makes the decision. Node $u$ terminates the algorithm when it decides on $decision_u$. Before termination, it broadcasts three messages for the next phase. If less than $2f$ messages but more than $f$ messages have the same value, then $u$ does not change $v_u$ and keeps $v_u$ in the next phase. If there is no value exists in more than $f$ messages, node $u$ takes a coin toss to get a random value and keeps this value in the next phase.
	
	This protocol solves the Byzantine agreement problem because it satisfies validity, exact agreement, and probabilistic termination. We discuss the correctness of this protocol in these three aspects.
	
	\begin{theorem}
		Algorithm 1 solves the Byzantine agreement problem in asynchronous distributed systems in the presence of $f<\frac{n}{3}$ Byzantine nodes.
	\end{theorem}
	
	\begin{lemma}
		Algorithm 1 satisfies the validity property.
	\end{lemma}
	
	\begin{proof}
	    There are two possible situations of node inputs in the system. In the first situation, all correct nodes have the same input value. In the second situation, correct nodes have different input values.
	
		If all correct nodes have the same input $v$, then all of them will receive more than $\frac{n-f}{2}>f$ messages of round 1 with the identical value $v$. In round 2, $v_u$ does not change because Byzantine nodes cannot create more than $\frac{n}{2}$ copies of malicious messages. For the same reason, all correct nodes will decide on $v$ in round 3.
		
		If nodes have different inputs, no matter which value they agree on, the validity property is always satisfied.
	\end{proof}
	
	\begin{lemma}
		Algorithm 1 satisfies the exact agreement property.
	\end{lemma}
	
	\begin{proof}
		First, we claim that two correct nodes $u$ and $w$ won't decide on different values in round $3k+3$. Suppose not: then node $u$ decides on 0 while node $w$ decides on 1 in round $3k+3$. Because both of them validate more than $2f$ messages, there are two correct nodes $u'$, and $w'$ get ready with value 0 and value 1 in round $3k+2$, i.e., $ready=TRUE$, respectively. Node $u'$ validates more than $\frac{n}{2}$ messages with value 0 in round $3k+2$ and node $u'$ validates more than $\frac{n}{2}$ messages with value 1 in round $3k+2$. There must be a node broadcasting two messages with value 0 and value 1 separately in round $3k+2$, which is impossible. Thus, correct nodes decide on the same value in the same round.
		
		Suppose that node $u$ decide on 0 in round $3k+3$, it validates $(2f+1)$ messages with value 0. Hence, other correct nodes validate at least $(f+1)$ messages of these $(2f+1)$ messages which are validated by $u$. Because node $u$ will continue broadcasting correct messages with value 0 in phase $k+1$. All correct nodes have the same value from round $3(k+1)+1$ and all of them will decide on 0 in round $3(k+1)+3$.
		
	\end{proof}
	
	\begin{lemma}
		Algorithm 1 satisfies the probabilistic termination property.
	\end{lemma}
	
	\begin{proof}
		
		Let us consider node actions in phase $k$. All nodes are still active in phase $k$. There are four possible situations that a correct node $u$ can be found in round $3k+3$.
		
		In the first situation, if node $u$ validates more than $2f$ messages with value $v$ in round $3k+3$, apparently, all correct nodes will decide on value $v$ in round $3(k+1)+3$ with probability 1.
		
		In the second situation, node $u$ validates more than $f$ messages with value $v$ in round $3k+3$. Other correct nodes will not decide on the value $v'\neq v$ or directly set its value to $v'$. Otherwise, with the same argument in the proof of Lemma 2, there must be a node broadcasting two messages with value 0 and value 1 respectively in round $3k+2$, which is impossible. Thus, the probability that other correct nodes start with value $v$ in phase $k+1$ is greater than $2^{-(n-f)}$.
		
		In the third situation, node $u$ validates less than $f$ messages with value $v$ in round $3k+3$. Other correct nodes might validate $v$ but won't validate a value other than $v$. The probability that all correct nodes start with value $v$ in phase $k+1$ is greater than $2^{-(n-f)}$.
		
		In the fourth situation, node $u$ has not validated a message with value $v$. Other correct nodes have not validated more than $f$ messages with value $v$. Otherwise, node $u$ will validate $v$ at least once. Therefore, all correct nodes toss coins and the probability that all correct nodes start with value $v$ in phase $k+1$ is $2^{-(n-f)}$.
		
		In all of these situations, algorithms terminates in phase $k+1$ with probability greater than or equal to $2^{-(n-f)}$. Thus, the probability that the algorithm never terminates is $\mathbb{P}(never\ terminating)\leq\lim_{k\rightarrow \infty }(1-2^{-(n-f)})^k=0$.
		
	\end{proof}
	
	\subsection{Byzantine Reliable Broadcast}
	
	In our model, we assume that the transmission mechanism between nodes is point-to-point, which does not meet the requirement of the top level protocol in Subsection \textit{A}; we use multiple \textbf{broadcast($\cdot$)} in Algorithm 1. We seek for a reliable broadcast to fix this gap. In this subsection, we discuss the details of a Byzantine reliable broadcast algorithm that helps us to implement the algorithm in our model.
	
	The Byzantine reliable broadcast allows nodes to broadcast in identical message $m$ to all nodes in the system. Nodes will validate the message as the response of the broadcast. A Byzantine reliable broadcast algorithm satisfies five properties\cite{cachin2011introduction},

	\textbf{Validity}: if a correct node $u$ broadcasts a message $m$, then every correct node eventually validates $m$.
	
    \textbf{No duplication}: every correct node validates message $m$ at most once in the broadcast of message $m$.
    
	\textbf{Integrity}: if a correct node $w$ validates a message $m$ from another correct node $u$, then $m$ was broadcast by $u$ before $w$ validating $m$.
	
	\textbf{Consistency}: if a correct node $w$ validates a message $m$ and another correct node $w'$ validates a message $m'$ in the same round, then $m=m'$.
	
	\textbf{Totality}: if a correct node $w$ validates a message $m$, then other correct nodes eventually validate $m$.

	We follow the authenticated double echo broadcast algorithm in\cite{bracha1987asynchronous}. The algorithm consists of three parts: broadcasting messages, echoing messages, and validating messages, which are shown in Algorithm 2, Algorithm 3, and Algorithm 4, respectively. Note that we assume full communication capability among nodes in this subsection, but it is not true in our model, we will explain how to solve this problem in Subsection \textit{C}, Subsection \textit{D} and Subsection \textit{E}.
	
	\begin{algorithm}[!htb]
		\caption{Broadcast Message $m=(\mbox{source} = u, \mbox{round} = k, \mbox{value} = v_u)$, Code for node $u$}
		\begin{algorithmic}[1]
			\State \textbf{Send} $(m, \mbox{from} = u, \mbox{label} = initial)$ to all nodes in the system.
		\end{algorithmic}
	\end{algorithm}
	
	Algorithm 2 indicates how node $u$ broadcasts a message $m$. It sends $m$ to all other nodes in the system with a $initial$ label.
	
	\renewcommand{\algorithmicfunction}{\textbf{upon}}
	\begin{algorithm}[!htb]
		\caption{Echo Message $m=(\mbox{source} = u, \mbox{round} = k, \mbox{value} = v_u)$, Code for node $w$}
		\begin{algorithmic}[1]
			\State $echo_w\gets\perp^N$
			\State $ready_w\gets\perp^N$
			\Function{${\rm accept\ message}\ (m, \mbox{from} = u, \mbox{label} = initial)$}{}
			\State \textbf{send} $(m, \mbox{from} = w, \mbox{label} = echo)$ to other nodes.
			\EndFunction
			\State
			\Function{${\rm accept\ message}\ (m=(\mbox{source} = u, \mbox{round} = k, \mbox{value} = v_u), \mbox{from} = v, \mbox{label} = echo)$}{}
			\If {$(\mbox{source} = u, \mbox{round} = k, ...)\notin echo_w[v]$}
			\State $echo_w[v]\gets echo_w[v]\cup m$
			\EndIf
			\EndFunction
			\State
			\Function{$\#(m\in echo_w[p], \forall p\in V_G)>\frac{n+f}{2}\ {\rm and}\ m\neq\perp$}{}
			\If{have not sent $(m, \mbox{from} = w, \mbox{label} = ready)$}
			\State \textbf{send} $(m, \mbox{from} = w, \mbox{label} = ready)$\\\ \ \ \ \ \ \ \ \ \ \ \ \ \ \ to other nodes.
			\EndIf
			\EndFunction
			\State
			\Function{${\rm accept\ message}\ (m=(\mbox{source} = u, \mbox{round} = k, \mbox{value} = v_u), \mbox{from} = v, \mbox{label} = ready)$}{}
			\If {$(\mbox{source} = u, \mbox{round} = k, ...)\notin ready_w[v]$}
			\State $ready_w[v]\gets ready_w[v]\cup m$
			\EndIf
			\EndFunction
			\State
			\Function{$\#(m\in ready_w[p], \forall p\in V_G)>f\ {\rm and}\ m\neq\perp$}{}
			\If{have not sent $(m, \mbox{from} = w, \mbox{label} = ready)$}
			\State \textbf{send} $(m, \mbox{from} = w, \mbox{label} = ready)$\\\ \ \ \ \ \ \ \ \ \ \ \ \ \ \ to other nodes.
			\EndIf
			\EndFunction
		\end{algorithmic}
	\end{algorithm}
	
	The algorithm is called authenticated double-echo broadcast because it has two echo steps, which are indicated in Algorithm 3. In the first step, node $w$ sends message $m$ with $echo$ label to other nodes when it accepts the message $m$ with $initial$ label from $u$. In the second step, $w$ sends message $m$ with $ready$ label to other nodes when it accepts more than $\frac{n+f}{2}$ copies of the message $m$ with $echo$ label from other nodes or when it accepts more than $f$ copies of the message $m$ with $ready$ label from other nodes.
	
	\begin{algorithm}[!htb]
		\caption{Validate Message $m=(\mbox{source} = u, \mbox{round} = k, \mbox{value} = v_u)$, Code for node $w$}
		\begin{algorithmic}[1]
			\State $Val\gets \emptyset$
			\Function{$\#(ready_w[p]=m, \forall p\in V_G)>2f\ {\rm \textbf{and}}\ m\neq\perp$}{}
			\If{$(\mbox{source} = u, \mbox{round} = k, ...)\notin Val$}
			\State \textbf{validate} $m$ as a message broadcast by $u$\\\ \ \ \ \ \ \ \ \ \ \ \ \ \ \ \ \ \ \ in round $k$
			\State $Val\gets Val\cup m$
			\EndIf
			\EndFunction	
		\end{algorithmic}
	\end{algorithm}
	
	We introduce the validating step in Algorithm 4. The message $m$ is validated by node $w$ when it accepts more than $2f$ copies of the message $m$ with $ready$ label from other nodes. This message is stored in the set $Val$. Set $Val$ is an empty set at the beginning of the algorithm.
	
	\begin{theorem}
		The authenticated double-echo broadcast algorithm is a Byzantine reliable broadcast.
	\end{theorem}
	
	\begin{proof}
	We prove Theorem 3 by indicating that the authenticated double-echo broadcast algorithm satisfies five properties of the Byzantine reliable broadcast in Lemma 4, Lemma 5, Lemma 6, Lemma 7 and Lemma 8.
	\end{proof}
	
	\begin{lemma}
		The authenticated double-echo broadcast algorithm satisfies the validity property.
	\end{lemma}
	\begin{proof}
		If a correct node $u$ broadcasts a message $m$, then other correct nodes will accept it and send a $echo$ message to others nodes in the system. Because $f<\frac{n}{3}$, we have $n-f>\frac{n+f}{2}$ and $f<\frac{n+f}{2}$. Correct nodes will accept at least $\frac{n+f}{2}+1$ $echo$ copies of message $m$ and send a $ready$ message to others.
		
		For the same reason, correct nodes will eventually accept at least $(2f+1)$ copies of message $m$ with $ready$ label from other correct nodes. Hence, message $m$ will be validated by every correct node eventually.
	\end{proof}
	
	\begin{lemma}
		The authenticated double-echo broadcast algorithm satisfies the no duplication property.
	\end{lemma}
	\begin{proof}
		The no duplication property is intuitive because $m$ is validated by node $w$ if it is not in the set $Val$. $m$ will be added to $Val$ after validation by node $w$. Hence every correct node only validates $m$ once.
	\end{proof}
	
	\begin{lemma}
		The authenticated double-echo broadcast algorithm satisfies the integrity property.
	\end{lemma}
	\begin{proof}
		If a correct node sends a message $m$ with $ready$ label to other nodes, it needs at least $\frac{n+f}{2}+1>f$ accepted $echo$ messages or $(f+1)$ accepted $ready$ messages. Byzantine nodes have no ability to realize these two events.
		
		A correct node needs at least $(2f+1)$ accepted $ready$ copies of a message $m$ to validate $m$. It is impossible for Byzantine nodes to create $(2f+1)$ fake messages with $ready$ label because there are less than $f$ Byzantine nodes in the system. Thus, message $m$ was broadcast from another node previously.
	\end{proof}
	
	\begin{lemma}
		The authenticated double-echo broadcast algorithm satisfies the consistency property.
	\end{lemma}
	\begin{proof}
		The integrity property of the authenticated double-echo broadcast algorithm implies that $m$ and $m'$ are all broadcast by node $u$ previously. Because node $u$ only broadcasts one message in a round, we must have $m=m'$. This message was broadcast by $u$ in this broadcast round.
	\end{proof}
	
	\begin{lemma}
		The authenticated double-echo broadcast algorithm satisfies the totality property.
	\end{lemma}
	\begin{proof}
		If a correct node $u$ validates a message $m$, it has already accepted at least $(2f+1)$ messages with $ready$ label. At least $(f+1)$ of them are sent by correct nodes.
		
		Thus, every correct node will accept at least $(f+1)$ messages with a $ready$ label from other correct nodes eventually. Because the authenticated double-echo broadcast algorithm requires a node to send a message of $m$ with $ready$ label when it accepts $(f+1)$ messages with $ready$ label.
		
		For this reason, every correct node will send a $ready$ message of $m$ to each other. Every correct node will validate $m$ eventually.
	\end{proof}
	
	\subsection{Send Messages}
	
	The authenticated double-echo broadcast algorithm was initially implemented in the system where nodes can communicate with each other directly. However, in our model, communication is restricted because the network is incomplete. The authenticated channels are only between adjacent nodes of the network. Danny Dolev introduced the Purifying algorithm, which has been used to solve communication problems on incomplete graphs\cite{dolev1982byzantine}. However, this algorithm only works in synchronous systems.
	
	In this subsection, Subsection \textit{D} and Subsection \textit{E}, we proposed an asynchronous purifying algorithm that fixes the communication gap in the system where not every pair of nodes can send and receive messages directly. The algorithm has the following properties such that the correctness of the authenticated double-echo broadcast algorithm is guaranteed.
	
	\textbf{Validity}: if a correct node $u$ sends a message $m$ to another correct node $w$, then node $w$ eventually accepts $m$.
	
	\textbf{No duplication}: if a correct node $u$ sends a message $m$ to another correct node $w$, then node $w$ only accepts $m$ once.
		
	\textbf{Integrity}: if a correct node $w$ accepts a message $m$ from another correct node $u$, then $m$ was broadcast by node $u$ before node $w$ accepting message $m$.
	
	We introduce our algorithm by giving an example of message delivery, i.e., node $u$ sends a message $m$ to node $w$ while node $v$ is an internal node of a path between $u$ and $w$. The asynchronous purifying algorithm consists of three parts: sending $m$ from the starting node $u$, transmitting $m$ via a path by an intermediate node $v$, and accepting $m$ at the destination node $w$, these steps are introduced respectively in this subsection, Subsection \textit{D} and Subsection \textit{E}.
	
	\begin{center}
		\begin{figure}[htbp]
			\centering
			\includegraphics[width=8.5cm]{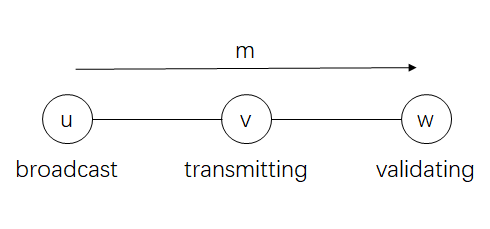}
			\caption{Node $u$ sends a message $m$ to node $w$ via a path including node $v$.}
		\end{figure}
	\end{center}

	Algorithm 5 demonstrates how node $u$ acts when $u$ decides to send a message $m$ to node $w$. As we discussed in Subsection \textit{B}, node $u$ sends a message $m$ with a label. When node $u$ sends a message $m$ to its neighbors, it sends the message $m$ together with the label, its own identifier $u$, and the message flooding path. Initially, the flooding path of the message is empty, which is represented by $\perp$.
	
	\begin{algorithm}[!htb]
		\caption{Send $(m, \mbox{from} = u, \mbox{label} = l)$ to node $w$, Code for node $u$}
		\begin{algorithmic}[1]
			\State \textbf{send} message $(m, \mbox{from} = u, \mbox{label} = l, \mbox{path} = \perp )$ to $u.neighbor$
		\end{algorithmic}
	\end{algorithm}
	
	\subsection{Transmit Messages}
	
	Because the communication graph is not complete, messages cannot be directly delivered to the destination. Internal nodes on paths between message source node and message destination node are responsible for transmitting messages. In our example, when node $v$ receives a message $m, \mbox{from} = u, \mbox{label} = l, \mbox{path} = \Pi$ from a neighbor $t$, it will send this message to its neighbors other than $t$ and store it in its local memory. The details are presented in Algorithm 6. 
	
	Node $v$ first checks that if the message flooding path $\Pi$ includes itself. If so, $v$ discards this message. This step ensures that node $v$ does not accept and send any duplicated message. Node $v$ discards a message if the message flooding path is incorrect. If the message is not discarded, then node $v$ stores this message, message label and the message flooding path in its local memory and sends $(\mbox{source} = u, \mbox{round} = k, \mbox{value} = v_u, \mbox{path} = \Pi-t)$ to all neighbors except $t$.
	
	\renewcommand{\algorithmicfunction}{\textbf{upon}}
	\begin{algorithm}[!htb]
		\caption{Transmit $(m, \mbox{from} = u, \mbox{label} = l, \mbox{path} = \Pi )$, Code for node $v$}
		\begin{algorithmic}[1]
			\Function{${\rm receive}\ (m, \mbox{from} = u, \mbox{label} = l, \mbox{path} = \Pi ) \ {\rm from}\ t$}{}
			\If{$v\in \Pi$}
			\State ignore the message
			\EndIf
			\If{$u\notin \Pi$}
			\State ignore the message
			\EndIf
			\State \textbf{store} $(m, \mbox{from} = u, \mbox{label} = l, \mbox{path} = \Pi-t)$\\\ \ \ \ \ \ \ \ \ \ \ \ in memory list $Mem$
			\State \textbf{send} $(m, \mbox{from} = u, \mbox{label} = l, \mbox{path} = \Pi-t)$\\\ \ \ \ \ \ \ \ \ \ \ \ to $v.neighbor\setminus t$
			\EndFunction
		\end{algorithmic}
	\end{algorithm}
	
	\subsection{Accept Messages}
	
	At node $w$, it accepts message $m$ if it receives enough copies of $m$ with the same label via disjoint flooding paths. The part of accepting messages is shown in Algorithm 7. Node $w$ accepts a message $m$ from node $u$ with label $l$ when it receives more than $f$ identical copies with label $l$ via disjoint flooding paths. This message is stored in $Acpt$. $Acpt$ is initialized as an empty set at the beginning of the whole algorithm.
	
	\begin{algorithm}[!htb]
		\caption{Accept $(m, \mbox{from} = u, \mbox{label} = l)$, Code for node $w$}
		\begin{algorithmic}[1]
			\State $Acpt\gets\emptyset$
			\Function{$\#((m, \mbox{from} = u, \mbox{label} = l, ...)\in Mem)>f$ \textbf{and} {\rm their} \textit{flooding paths} {\rm are disjoint}}{}
			\If {$(m, \mbox{from} = u, \mbox{label} = l)\notin Acpt$}
			\State \textbf{accept} $m$ as a message from node $u$ with label $l$ 
			\State $Acpt\gets Acpt\cup (m, \mbox{from} = u, \mbox{label} = l)$
			\EndIf
			\EndFunction
		\end{algorithmic}
	\end{algorithm}
	
	\begin{theorem}
		The asynchronous purifying algorithm satisfies the validity, no duplication, and integrity properties, which guarantees the correctness of the authenticated double-echo broadcast algorithm in our model.
	\end{theorem}
	
	\begin{proof}
	
	The no duplication property is intuitive because $(m, \mbox{from} = u, \mbox{label} = l)$ is accepted if it is not in $Acpt$.
	
	We prove the correctness of this asynchronous purifying algorithm from two perspectives. We fist prove the validity property in Lemma 9 by showing that every message $m$ sent by node $u$ will be accepted by node $w$ eventually. Then we prove that if a message $m$ is accepted by node $w$, it must be sent by node $u$ first to verify the integrity property in Lemma 10.
	
	\end{proof}

	\begin{lemma}
		When a correct node $u$ broadcasts a message $m$ with label $l$, another correct node $w$ will receive at least $(f+1)$ correct copies via disjoint flooding paths if the communication graph $G$ has $(2f+1)$ vertex connectivity.
	\end{lemma}
	
	\begin{proof}
		Because the connectivity of the communication graph $G$ is $(2f+1)$, there are at least $(2f+1)$ disjoint paths between any node pair. Therefore, there are at least $(2f+1)$ disjoint paths between the source node $u$ and the destination node $w$.
		
		There are at most $f$ byzantine nodes in the system. These nodes can only appear in at most $f$ disjoint paths between $u$ and $w$. Thus, there are at least $(f+1)$ disjoint paths only contains correct nodes between $u$ and $w$.
		
		Correct messages will be transmitted through these paths from node $u$ to node $w$. Hence, node $w$ will receive at least $(f+1)$ correct copies via disjoint flooding paths, and message $m$ with label $l$ from node $u$ will be accepted by node $w$ eventually.
	\end{proof}
	
	\begin{lemma}
		Every message $m$ with label $l$ which is validated in the asynchronous purifying algorithm is a correct message.
	\end{lemma}
	
	\begin{proof}
		
		We prove that Byzantine nodes cannot manipulate validated messages.
		
		Consider $s$ is the last Byzantine node in a path $P$ between node $u$ and node $w$. We claim that if a message $m$ with label $l$ is delivered to $w$ successfully through $P$, then $s$ must appear in the flooding path of this message because node $s$ is recorded in the flooding path by the next correct node after node $s$ in $P$. This information will not be modified by other nodes in the graph because $s$ is the last Byzantine node in $P$.
		
		Byzantine nodes appear in at most $f$ disjoint paths. Thus, at most $f$ copies of $m$ with label $l$ can be modified by them. However, we need $(f+1)$ identical copies to validate message $m$ with label $l$. Therefore, there must be at least one copy of $m$ is correct. This message is transmitted through a path without any Byzantine node.
		
		There is at least one correct copy of $m$ with label $l$ accepted by node $w$ when $m$ with label $l$ is validated by node $w$. Hence, message $m$ is a correct message sent by node $u$ if it is validated in the asynchronous purifying algorithm.
	\end{proof}
	
	Until now, we accomplished the randomized protocol that solves the Byzantine agreement problem with the presence of $f$ Byzantine nodes in our model.

	\begin{theorem}
		If the communication graph $G$ of the asynchronous distributed system with $n$ nodes and $f<\frac{n}{3}$ Byzantine nodes have $(2f+1)$ vertex connectivity, then our algorithm solves the asynchronous Byzantine agreement problem, that all correct nodes decide on the same value with high probability.
	\end{theorem}
	
	\section{Necessary Conditions}
	
	In the previous section, we presented a randomized algorithm that solves the Byzantine agreement problem with $n$ nodes with the presence of $f$ Byzantine nodes in our model where the connectivity of the communication graph is at least $(2f+1)$. This algorithm gives us a upper bound of the vertex connectivity requirement for the communication graph.
	
	In this section, we prove that this bound is tight. There does not exist an algorithm that can solve the asynchronous Byzantine agreement problem with high probability when the vertex connectivity of $G$ is less than $(2f+1)$.
	
	\begin{theorem}
		If the communication graph $G$ of the asynchronous distributed system with $n$ nodes and $f<\frac{n}{3}$ Byzantine nodes does not have $(2f+1)$ vertex connectivity, then it is impossible to find an algorithm that solves the asynchronous Byzantine agreement problem, i.e., all correct nodes decide on the same value.
	\end{theorem}
	
	There are two different requirements in Theorem 6. We prove Theorem 6 with Lemma 11 and Lemma 12.
	
	\begin{lemma}[\cite{bracha1983resilient}]
		If there exists an algorithm that solves the Byzantine agreement problem in our model with $n$ nodes with the presence of $f$ Byzantine nodes, then $n\geq 3f+1$.
	\end{lemma}
	
	\begin{lemma}
		If there exists an algorithm that solves the Byzantine agreement problem in our model with $n$ nodes on the communication graph $G$ in the presence of $f$ Byzantine nodes, then $G$ is $(2f+1)$-connected.
	\end{lemma}
	
	Lemma 11 was proved in \cite{bracha1983resilient}. It is impossible for randomized algorithms to solve the asynchronous Byzantine agreement problem if the number of Byzantine nodes $f$ is greater than or equal to $\frac{n}{3}$.
	
	We prove Lemma 12 by using the state machine based approach, which is similar to the technique in \cite{fischer1986easy, bracha1983resilient, dolev1986reaching}. We discuss an example that it is impossible for any randomized algorithm to solve the Byzantine agreement problem with high probability. Correct nodes will decide on different values.
	
	\begin{proof}[Proof for Lemma 12]

		Suppose for the sake contradiction, there exists an algorithm $A$ that solves the asynchronous Byzantine agreement problem with high probability in our model with $n$ nodes, tolerating at most $f< \frac{n}{3}$ Byzantine faulty nodes, communicating on the communication graph $G = (V_G, E_G)$, where $G$ is not $(2f+1)$-connected. $A$ outlines a procedure $A_u$ for each node $u\in V_G$ that describes state transitions of $u$.
		
		We can find a vertex cut of $G$ and the size of the vertex cut is less than $(2f+1)$. Let $C$ be the vertex cut of $G$ and $ \left | C \right |\leq 2f$. Other vertices $V_G\setminus C$ are partitioned into two non-empty vertex sets $X, Y$, such that $X$ and $Y$ are disconnected in $G\setminus C$.
		
		Because $ \left | C \right |\leq 2f$, we can partition $C$ into two disjoint vertex sets $R$ and $T$, i.e., $C = R\cup T$ and $0< \left | R \right |\leq f$ and $0<\left | T \right |\leq f$. The structure of network $G$ is shown in Figure 4.
		
		\begin{center}
			\begin{figure}[htbp]
				\centering
				\includegraphics[width=5cm]{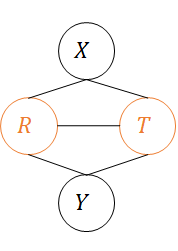}
				\caption{Network structure of $G$.}
			\end{figure}
		\end{center}
		
		Let us consider three executions $E_1, E_2,$ and $E_3$ on $G$. We run algorithm $A$ in these three executions.
		
		In the first execution $E_1$, nodes in $T$ are Byzantine while other nodes in $X, Y, R$ are correct with input 0.
		
		In the second execution $E_2$, nodes in $T$ are Byzantine, and other nodes in $X, Y, R$ are correct with input 1.
		
		In the third execution $E_3$, nodes in $X, Y, T$ are correct, while nodes in $R$ play Byzantine strategies. We input 0 to nodes of $X$ and input 1 to nodes in $Y$ and $T$.
		
		It is easy for us to figure out the outputs of nodes in $E_1$ and $E_2$. Because of the validity of $A$, all correct nodes in $E_1$ decide on 0, no matter how Byzantine nodes act. For the same reason, all correct nodes in $E_2$ agree on 1.
		
		However, it is not intuitive what is the output of nodes in the third execution $E_3$. We will explain that correct nodes in $E_3$ will output different values if Byzantine nodes in $R$ play specific strategies. To understand the behavior of nodes in $E_3$, we introduce the fourth execution $E_4$ on another graph $H$.
		
		We construct a network $H=(V_H, E_H)$ based on $G$. For each node $u\in V_G$, there are two copies of $u$, i.e., $u_0, u_1\in V_H$. Thus, $V_H$ can be partitioned into 8 vertex sets: $(X_0, Y_0, R_0, T_0, X_1, Y_1, R_1, T_1)$.
		
		If $(u,v)\in E_G$ and $u, v$ are in the same vertex set, then we copy $(u, v)$ twice in $E_H$, such that $(u_0, v_0) \in E_H, (u_1,v_1) \in E_H$. If $(u,v)\notin E_G$, then there is no edge between $(u_0, v_0)$ and $(u_1, v_1)$.
		
		If $(u,v)\in E_G$ and $u, v$ are not in the same vertex set, we build connections according to the following rules.
		
		\begin{itemize}
			\item If $u\in R, v\in T$ and $(u,v)\in E_G$, then $(u_0, v_0)\in E_H$, $(u_1, v_1)\in E_H$.
			\item If $u\in X, v\in R$ and $(u,v)\in E_G$, then $(u_0, v_0)\in E_H$, $(u_1, v_1)\in E_H$.
			\item If $u\in Y, v\in R$ and $(u,v)\in E_G$, then $(u_0, v_0)\in E_H$, $(u_1, v_1)\in E_H$.
			\item If $u\in Y, v\in T$ and $(u,v)\in E_G$, then $(u_0, v_0)\in E_H$, $(u_1, v_1)\in E_H$.
			\item If $u\in X, v\in T$ and $(u,v)\in E_G$, then $(u_1, v_0)\in E_H$, $(u_0, v_1)\in E_H$.
		\end{itemize}

		We give the structure of $H$ in Figure 5. Edges within vertex sets are not displayed, while edges between vertex sets are represented by a single edge. 
		
		\begin{center}
			\begin{figure}[htbp]
				\centering
				\includegraphics[width=8.5cm]{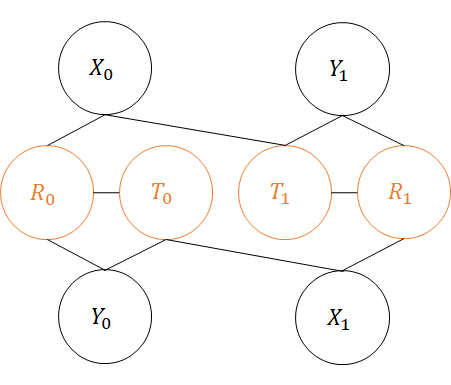}
				\caption{Network structure of $H$ to run execution $E_4$.}
			\end{figure}
		\end{center}
		
		$E_4$ is an execution on $H$ as follows. Each node pair $u_0, u_1\in V_H$ runs $A_u$, which is the same algorithm protocol as $u\in V_G$ runs. Nodes in $X_0, Y_0, R_0, T_0$ have initial input 0 and nodes in $X_1, Y_1, R_1, T_1$ have initial input 1. All nodes are correct.
		
		To understand the output of nodes in $E_3$, we first discuss the output of nodes in $E_4$. We start from nodes in $X_0, Y_0, R_0$ in $E_4$. We claim that the behavior of $X_0, Y_0, R_0$ in $E_4$ is modelled by $X, Y, R$ in $E_1$.
		
		Let us consider strategies of Byzantine nodes in $T$ in $E_1$. Each node $u\in T$ in $E_1$ can play a mixed strategy of the actions of $(u_0\in T_0, u_1\in T_1)$ in $E_4$. It considers itself as a combination of $u_0$ and $u_1$ in $E_4$. It reacts to nodes in $X$ as the same as the reaction of $u_1$ to nodes in $X_0$ and reacts to nodes in $Y$ and $R$ as the same as the reaction of $u_0$ to nodes in $Y_0, R_0$.
		
		As we discussed before, no matter how nodes in $T$ behavior, the output of nodes in $X, Y, R$ in $E_1$ is 0. Because nodes in $X_0, Y_0, R_0$ in $E_4$ are in the same environment, have the same input, and run the same algorithm as nodes in $X, Y, R$ in $E_1$, they will also decide on 0 eventually. This simulation is shown in Figure 6.

		\begin{center}
			\begin{figure}[!htb]
				\centering
				\includegraphics[width=6cm]{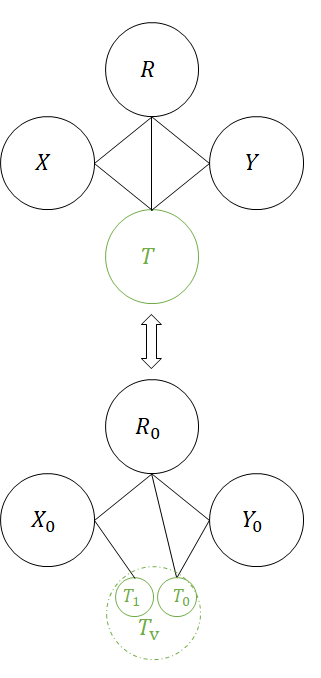}
				\caption{The behavior of nodes in $E_1$ and nodes in $E_4$. Green nodes are Byzantine nodes in $E_1$. $T_v$ copies behavior of $T_1$ to $X_0$ and $T_0$ to $Y_0, R_0$. Nodes in $X, Y, R$ decide on 0 eventually in $E_1$ while nodes in $X_0, Y_0, R_0$ decide on 0 eventually in $E_4$.}
			\end{figure}
		\end{center}
		
		For the same reason, nodes in $X_1, Y_1, R_1$ decide on 1 eventually in $E_4$. We show the simulation of nodes in $X_1, Y_1, R_1$ in $E_4$ in Figure 7.

		\begin{center}
			\begin{figure}[!htb]
				\centering
				\includegraphics[width=6cm]{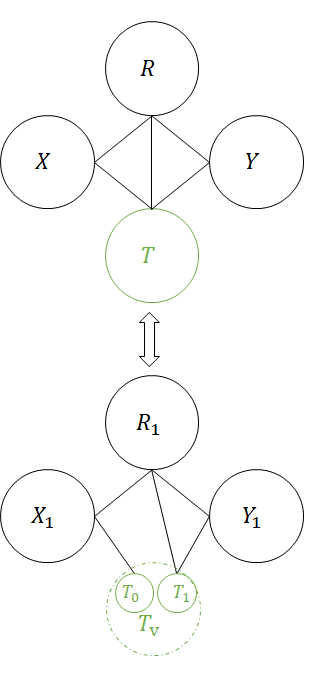}
				\caption{The behavior of nodes in $E_2$ and nodes in $E_4$. Green nodes are Byzantine nodes in $E_2$. $T_v$ copies behavior of $T_0$ to $X_1$ and $T_1$ to $Y_1, R_1$. Nodes in $X, Y, R$ deiced on 1 with high probability in $E_2$ while nodes in $X_1, Y_1, R_1$ deiced on 1 with high probability in $E_4$.}
			\end{figure}
		\end{center}
		
		Each Byzantine node $u\in T$ in $E_2$ considers itself as a combination of $u_0$ and $u_1$ in $E_4$. It reacts to nodes in $X$ as the same as the reaction of $u_0$ to nodes in $X_1$ and reacts to nodes in $Y$ and $R$ as the same as the reaction of $u_1$ to nodes in $Y_1, R_1$. No matter how nodes in $T$ behave, the output of nodes in $X, Y, R$ in $E_2$ is 1. Because nodes in $X_1, Y_1, R_1$ in $E_4$ are in the same environment, have the same input and run the same algorithm as nodes in $X, Y, R$ in $E_2$, they will also decide on 1 eventually.
		
		Based on previous analysis, we can design Byzantine strategies for nodes in $R$ in $E_3$. Then nodes in $X$, $Y$, and $T$ will decide on different values eventually. Each node $u\in R$ plays a mixed strategy of actions of $u_0\in R_0$ and $u_1\in R_1$ in $E_4$. It reacts to nodes in $X$ as the same as the reaction of $u_0$ to nodes in $X_0$ and reacts to nodes in $Y$ and $T$ as the same as the reaction of $u_1$ to nodes in $Y_1, T_1$. This simulation is shown in Figure 8.
		
		\begin{center}
			\begin{figure}[!htb]
				\centering
				\includegraphics[width=6cm]{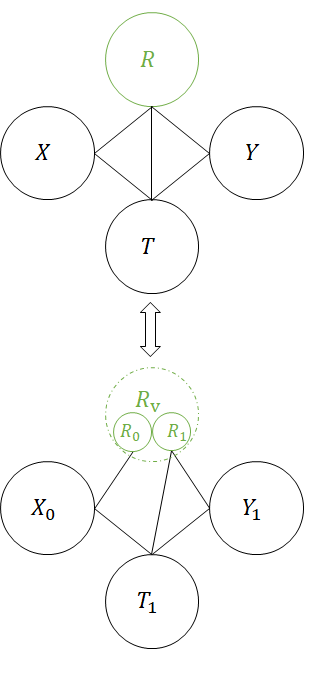}
				\caption{The behavior of nodes in $E_3$ and nodes in $E_4$. Green nodes in $R$ are Byzantine nodes in $E_3$, which copy the behavior of $R_0$ and $R_1$ in $E_4$. correct nodes in $X$ deiced on 0 and nodes in $Y$ decide on 1 with high probability in $E_3$.}
			\end{figure}
		\end{center}
		
		$X_0, Y_1, T_1$ in $E_4$ and $X, Y, T$ in $E_3$ have the same input, the same algorithm procedure, and the same environment. Hence, nodes in $X$ decide on 0, and nodes in $Y$ decide on 1 eventually in $E_3$. This result indicates that $A$ cannot achieve the Byzantine agreement in $E_3$, which contradicts the agreement property of $A$ and our assumption.
		
	\end{proof}
	
	In Lemma 11 and Lemma 12, we prove that $(2f+1)$ vertex connectivity is a necessary condition for solving the asynchronous Byzantine agreement problem in an incomplete network with $n$ nodes in the presence of $f<\frac{n}{3}$ Byzantine nodes.

	\section{Conclusion}
	
	In this work, we investigate the Byzantine agreement problem in the asynchronous distributed system with restricted communication. 
	
	Compare to previous work, we are the first to study the classical Byzantine agreement problem in a more realistic problem in world-scale distributed systems, i.e., blockchain systems.
	
	We prove that following conditions are necessary and sufficient to achieve Byzantine agreement among $n$ nodes with the presence of $f$ Byzantine nodes: the communication graph has $(2f+1)$-vertex connectivity and the number of Byzantine nodes $f$ is less than $\frac{n}{3}$. We also present a randomized algorithm that solves the exact Byzantine agreement problem on incomplete graphs in asynchronous systems.
	
	Beyond these contribution, we propose a three-layer framework that solves the asynchronous Byzantine agreement problem in incomplete networks. The bottom layer is a strong protocol which allows other algorithms \cite{mostefaoui2014signature} which solve the asynchronous Byzantine agreement problem to be implemented in incomplete networks.
	
	One might hope that some classic synchronizers can help. However, classic synchronizers cannot deal with Byzantine nodes. There are algorithms \cite{lamport1985synchronizing} that solve the Byzantine synchronization, however, only in complete graphs, with additional assumptions.

    To build a byzantine-tolerant synchronizer in an incomplete network seems to be an interesting problem, and we believe that our paper could be a stepping stone to understanding that problem. Though we cannot claim that we directly solve the problem, because the Byzantine clock synchronization problem can be considered as a Byzantine agreement problem with specified conditions.

\bibliographystyle{ACM-Reference-Format}
\bibliography{sample-base}

\end{document}